\documentclass[a4paper,english]{lipics-v2016}
\usepackage{amsfonts}

\newcommand{\gilt}{:}
\newcommand{\sodass}{\,:\,}
\newcommand{\setGilt}[2]{\left\{ #1\sodass #2\right\}}

\newcommand{\realrange}[2]{\left[#1, #2\right]}

\newcommand{\unitrange}[2]{\realrange{0}{1}}

\newcommand{\llabel}[1]{\label{\labelprefix:#1}}
\newcommand{\labelprefix}{} 

\newcommand{\discussionsize}{\small}

\newcommand{\frage}[1]{}

\newenvironment{code}{\noindent
\begin{tabbing}%
\hspace{2em}\=\hspace{2em}\=\hspace{2em}\=\hspace{2em}\=\hspace{2em}\=%
\hspace{2em}\=\hspace{2em}\=\hspace{2em}\=\hspace{2em}\=\hspace{2em}\=%
\kill}{\end{tabbing}}

\newcommand{\labelcommand}{}
\newcommand{\captiontext}{}
\newsavebox{\codeparam}
\newcounter{lineNumber}
\newenvironment{disscodepos}[3]{%
\renewcommand{\labelcommand}{#2}%
\renewcommand{\captiontext}{#3}%
\sbox{\codeparam}{\parbox{\textwidth}{#3}}%
\begin{figure}[#1]\begin{center}\begin{code}\setcounter{lineNumber}{1}}{%
\end{code}\end{center}\caption{\llabel{\labelcommand}\captiontext}\end{figure}}

{\end{disscodepos}}

\newcommand{\Is}       {:=}

\newdimen\endofsize\endofsize=0.5em
\def\endofbeweis{~\quad\hglue\hsize minus\hsize
                 \hbox{\vrule height \endofsize width
\endofsize}\par}

\usepackage[T1]{fontenc}
\usepackage[utf8]{inputenc}
\usepackage{color}
\definecolor{note_fontcolor}{rgb}{0.800781, 0.800781, 0.800781}
\usepackage{multirow}
\usepackage{amsmath}
\usepackage{amssymb}
\usepackage{mathtools}
\usepackage{babel}
\usepackage{url}
\usepackage{wrapfig}
\usepackage{numprint}
\usepackage{tikz}
\usepackage{pgfplots}
\pgfplotsset{compat=1.12}

\newcommand{\csch}[1]{{\color{blue}[CS: #1]}}
\newcommand{\mpopp}[1]{{\color{green}[MP: #1]}}
\renewcommand{\csch}[1]{}
\renewcommand{\mpopp}[1]{}

\def\MdR{\ensuremath{\mathbb{R}}}

\newcommand{\ie}{i.e.\ }

\newcommand{\eg}{e.g.\ }

\usepackage[left,pagewise] {lineno}
\definecolor{infocolor}{rgb} {0.6,0.6,0.6}

\title{Graph Partitioning with Acyclicity Constraints}
\author[1]{Orlando Moreira}
\author[1]{Merten Popp}
\author[2]{Christian Schulz}
\affil[1]{Intel Corporation, Eindhoven, The Netherlands, \\\url{orlando.moreira@intel.com}, \url{merten.popp@intel.com}}
\affil[2]{Karlsruhe Institute of Technology, Karlsruhe, Germany, and\\ University of Vienna, Vienna, Austria\\
  \texttt{christian.schulz@\{kit.edu, univie.ac.at\}}}
\begin{document}

\maketitle
\begin{abstract}
    Graphs are widely used to model execution dependencies in applications.
    In particular, the NP-complete problem of partitioning a graph under
    constraints receives enormous attention by researchers because of its
    applicability in multiprocessor scheduling. We identified the additional
    constraint of acyclic dependencies between blocks when mapping computer
    vision and imaging applications to a heterogeneous embedded multiprocessor.
    Existing algorithms and heuristics do not address this requirement and
    deliver results that are not applicable for our use-case. In this
    work, we show that this more constrained version of the graph
    partitioning problem is NP-complete and present heuristics that achieve
    a close approximation of the optimal solution found by an exhaustive search
    for small problem instances and much better scalability for larger
    instances.
    In addition, we can show a positive impact on the schedule of a real imaging
    application that improves communication volume and execution time.
\end{abstract}
\section{Practical Motivation}
\label{intro}

The context of this research is the development of computer vision and imaging
applications at Intel Corporation. These applications have high demands for
computational power but often need to run on embedded devices with severely
limited compute resources and a tight thermal budget. Our target platform is a
heterogeneous multiprocessor for advanced imaging and computer vision and is
currently used in Intel processors. It is designed for low power and has small
local program and data memories. To cope with the memory constraints, the
application developer currently has to \emph{manually} break the application, which is given as
a directed graph, into smaller blocks that are executed one after another. The
quality of this partitioning has a strong impact on communication volume and
performance. However, for large graphs this is a non-trivial task that requires
detailed knowledge of the hardware. Hence, the task should be handled by a
well-designed algorithm instead.

There are many existing heuristics for partitioning graphs into blocks of nodes
of roughly equal size. However, our platform has the requirement that there
must not be a cycle in the dependencies between the blocks because they have to
be executed one after another.

The contributions of this work are the identification of a new variation of the
graph partitioning problem, proofs it is NP-complete and hard
to approximate, as well as the implementation and evaluation of heuristics that address
this problem.
First, we present all necessary background information on the application graph
and hardware and explain our additional constraint in Section~\ref{background}.
We then continue to briefly introduce all basic concepts and related work in
Section~\ref{preliminaries}. We have not been able to identify works that
address the aforementioned constraint of our problem variant, which originates
from the hardware platform and disallows cycles between blocks.
The proofs are found in Section~\ref{proof} and the proposed heuristic
algorithms in Section~\ref{heuristics}.
We perform a number of experiments in Section~\ref{evaluation}, where we
use small graphs to compare our heuristics against an optimal algorithm that
uses exhaustive search.
We then evaluate our heuristics for larger graphs in the context of a
real-world imaging application and estimate the impact on the application.
In addition, we demonstrate the scalability of our heuristics with a set of
large graphs.
Finally, we conclude in Section~\ref{conclusion}.

\section{Background}
\label{background}
Computer vision and imaging applications can often be expressed as graphs where
nodes represent imaging functions and edges denote data dependencies. The
widely-accepted industry standard OpenVX \cite{openvx} released in 2014 by the
Khronos group uses a graph-based execution model. With the OpenVX API the
developer can specify the data flow of the application as a graph independent of
hardware constraints. The hardware vendor on the other hand can provide an API
implementation that uses advanced optimizations \cite{rainey2014addressing} like
specialized hardware, parallelized and pipelined node execution, overlapped
computation and data transfers and aggregated data transfers to avoid a round
trip to external memory.

\begin{figure}
\includegraphics[width=\columnwidth]{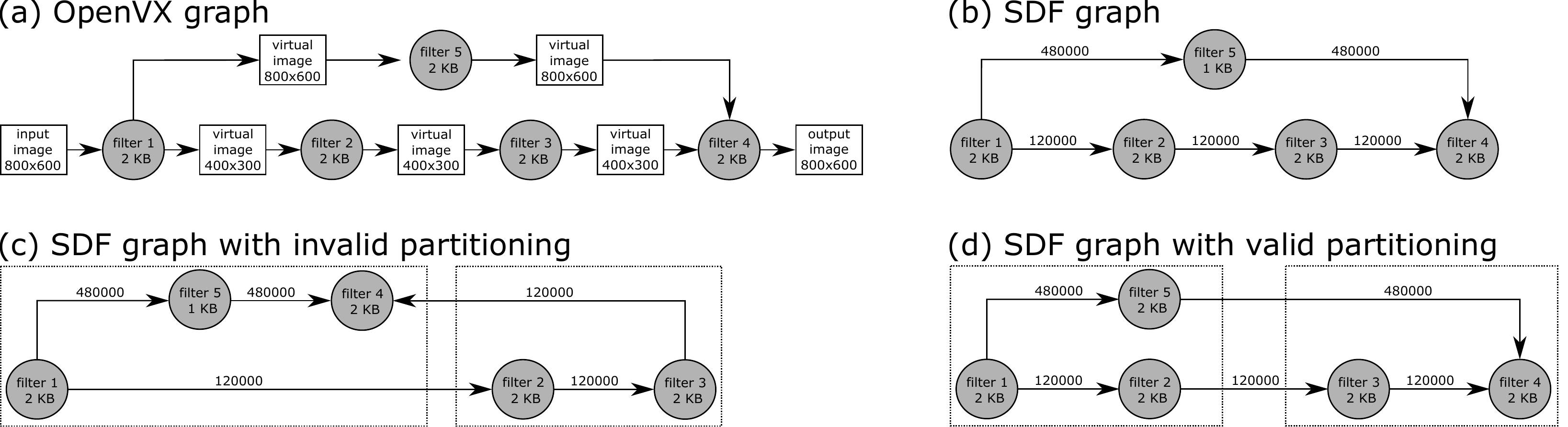}
\caption{Depiction of an imaging application graph. The nodes represent
processing kernels and are annotated with the size of the program binaries. The
OpenVX graph as specified by the developer using the API standard is shown in
(a), virtual images are intermediate data. The initial SDF representation is
shown in (b) where the edges are annotated with the buffer size, (c) shows an
invalid partition with minimal edge cut, but a bidirectional connection between
blocks and thus a cycle in the quotient graph. A valid partitioning with minimal
edge cut is shown in (d).\label{fig:OpenVX_example}}
\end{figure}

The application is specified as a Directed Acyclic Graph (DAG) in OpenVX\@. The
nodes of the DAG are either kernels (small, self-contained functions) or data
objects. Edges denote data dependencies and always connect exactly one kernel
with one data object. No cycles or feedback loops are allowed \cite{openvx}.
The need of some imaging algorithms to access previous data (\eg{}video
stabilization) is addressed by special OpenVX delay objects that hold data of
previous graph executions. Our existing tool flow converts the OpenVX graph in
linear time into an Synchronous Dataflow (SDF) graph.
SDF is a model of computation that abstracts from functionality and enables
several prevalent analysis and scheduling techniques~\cite{lee1987synchronous}.
In this representation, nodes represent processing
and the directed edges represent FIFO buffers.
The SDF nodes are annotated with the program size for each kernel in the OpenVX
graph.
If two kernels are linked by a data object in the OpenVX graph, the SDF nodes
are connected by a directed edge annotated with the size of the data object.
The resulting graph is a DAG\@.
An example of this conversion is shown~in~Figure~\ref{fig:OpenVX_example}a and
\ref{fig:OpenVX_example}b.

In this work, we address a graph partitioning problem that arises
when mapping the nodes of a DAG to the processing elements of a heterogeneous
embedded multiprocessor. The processing elements (PEs) of this platform have a
private local data memory and a separate program memory. A direct
memory access controller is used to transfer data between the local
memories and the external DDR memory~of~the~system.
The data memories have a size in the order of hundreds of kilobytes and can thus
only store a small portion of the image. Therefore the input image is divided
into \emph{tiles}. The mode of operation of this hardware usually is that the
nodes in the application graph are assigned to PEs and process the tiles one
after the other. In most cases this can be pipelined such that while the PEs
process the current tile, the direct memory access controller concurrently
loads the next tile to the local memories and writes the processed tile from the
previous iteration back to main memory.

However, this is only possible if the program memory size of the PEs is
sufficient to store all kernel implementations. For the hardware platform under
consideration it was found that this is not the case for more complex
applications such as a Local Laplacian filter~\cite{paris2011local}. Therefore a gang
scheduling \cite{feitelson1992gang} approach is used where the kernels
are divided into groups of kernels (referred to as gangs) that do not
violate memory constraints. Gangs are executed one after another on the target
platform. After each execution, the kernels of the next gang are loaded. At no
time any two kernels of different gangs are loaded in the program memories of
the processors at the same time.
Thus all intermediate data that is produced by the current gang but is needed by
a kernel in a later gang needs to be transferred~to~external~memory.

Since memory transfers, especially to external memories, are expensive in terms
of power, the assignment of nodes to gangs is crucially important. There are
many graph partitioning algorithms for the problem of dividing a graph into
blocks of nodes under certain conditions~\cite{SPPGPOverviewPaper}. However, in
this case we require a strict ordering of gangs.
Data objects may only be consumed in the same gang where they were produced and
in gangs that are scheduled later. If this does not hold, there is no valid
order in which the gangs can be executed on the platform. A valid order is a
topological ordering of the graph that represents data dependencies between
gangs. This graph can be created by taking the original DAG of the application
and contracting all nodes that are assigned to the same gang into a single node.
In order for a valid gang execution order to exist, the resulting graph
therefore must be a DAG itself.
An example for an incorrect assignment is shown
in~Figure~\ref{fig:OpenVX_example}c and a correct assignment
in~Figure~\ref{fig:OpenVX_example}d.

\section{Preliminaries}
\label{preliminaries}
In this section, we introduce the mathematical notation used throughout this
paper, give the formal definition of the graph partitioning problem and
show its relation to multiprocessor scheduling as a whole.

\subsection{Basic Concepts}
Let $G=(V=\{0,\ldots, n-1\},E,c,\omega)$ be an directed graph
with edge weights $\omega: E \to \MdR_{>0}$, node weights
$c: V \to \MdR_{\geq 0}$, $n = |V|$, and $m = |E|$.
We extend $c$ and $\omega$ to sets, i.e.,
$c(V')\Is \sum_{v\in V'}c(v)$ and $\omega(E')\Is \sum_{e\in E'}\omega(e)$.
We are looking for \emph{blocks} of nodes $V_1$,\ldots,$V_k$
that partition $V$, i.e., $V_1\cup\cdots\cup V_k=V$ and $V_i\cap V_j=\emptyset$
for $i\neq j$.
We call a block $V_i$ \emph{underloaded} [\emph{overloaded}] if
$c(V_i) < L_{\max}$ [if $c(V_i) > L_{\max}$]. If a node $v$ has a neighbor in
a block different of its own block then both nodes are called \emph{boundary
nodes}.
An abstract view of the partitioned graph is the so-called \emph{quotient
graph}, in which nodes represent blocks and edges are induced by connectivity
between blocks.
The \emph{weighted} version of the quotient graph has node weights which are set
to the weight of the corresponding block and edge weights which are equal to the
weight of the edges that run between the respective blocks.

\subsection{Problem Definition}
The partitions that we are looking for have to satisfy two constraints: a
balancing constraint and a acyclicity constraint. The \emph{balancing
constraint} demands that
$\forall i\in \{1..k\}\gilt c(V_i) \leq L_{\max} := (1+\epsilon)\lceil\frac{c(V)}{k}\rceil$
for some imbalance parameter $\epsilon \geq 0$.
The \emph{acyclicity constraint} mandates that
the quotient graph is acyclic.
The objective is to minimize the total \emph{cut} $\sum_{i,j}w(E_{ij})$ where
$E_{ij}\Is\setGilt{(u,v)\in E}{u\in V_i,v\in V_j}$.
The \emph{directed graph partitioning problem with acyclic quotient graph (DGPAQ)}
is then defined as finding a partition $\varPi:=\left\{V_{1,}\ldots,V_{k}\right\}$
that satisfies both constraints while minimizing the objective function.

\subsection{Relation to Scheduling}
The balancing constraint ensures that the size of the programs in a
scheduling gang does not exceed the program memory size of the platform
and thus is an important constraint for scheduling.
Reducing the edge cut reduces the amount of data transfers between gangs
and thus improves the memory bandwidth requirements of the application.
Note that an application is either compute-limited (processors are always
occupied) or bandwidth-limited (processors wait for data). Thus a minimization
of transfers does not guarantee an optimal schedule.
However, especially in embedded systems, the memory bandwidth is often the
bottleneck and a schedule requiring a large amount of transfers will neither
yield a good throughput nor good energy efficiency~\cite{panda2001data}.
Therefore, we address the problem of minimizing the edge cut under the
given constraints in isolation and do not solve a scheduling problem in this work.
We provide linear-time heuristics that can later be employed as subroutines in
broader scheduling algorithms to reduce data transfers.

\subsection{Related Work}
\label{relatedwork}
There has been a vast amount of research on graph partitioning so that we
refer the reader to \cite{schloegel2000gph,GPOverviewBook,SPPGPOverviewPaper}
for most of the material.
Here, we focus on issues closely related to our main contributions.
All general-purpose methods that are able to obtain good partitions for large
real-world graphs are based on the multilevel principle. The basic idea can be
traced back to multigrid solvers for systems of linear equations
\cite{Sou35} but more recent practical methods are based on mostly graph
theoretical aspects, in particular edge contraction and local search.
There are many ways to create graph hierarchies such as matching-based schemes
\cite{Walshaw07,karypis1998fast,Scotch} or variations thereof \cite{Karypis06}
and techniques similar to algebraic multigrid, \eg{}\cite{meyerhenke2006accelerating}.
We refer the interested reader to the respective papers for more details.
Well-known software packages based on this approach include
Jostle~\cite{Walshaw07}, KaHIP~\cite{kaffpa}, Metis~\cite{karypis1998fast} and
Scotch~\cite{ptscotch}. However, none of these tools can partition directed
graphs under the constraint that the quotient graph is a DAG.\@ We are not aware
of any related work that is able to~satisfy~this~constraint.

Gang scheduling was originally introduced to efficiently schedule parallel
programs with fine-grained interactions~\cite{feitelson1992gang}.
In recent work, this concept has been applied to schedule parallel applications
on virtual machines in cloud computing~\cite{stavrinides2016scheduling}
and extended to include hard real-time tasks~\cite{goossens2016optimal}.
An important difference to our work is that in gang scheduling all tasks that
exchange data with each other are assigned to the same gang, thus there is no
communication between gangs. In our work, the limited program memory of embedded
platforms does not allow to assign all kernels to the same gang.
Therefore, there is communication between gangs which we aim to minimize by
employing graph partitioning methods.

\section{Hardness Results}
\label{proof}
\begin{wrapfigure}{r}{0.45\textwidth}
\vspace*{-.5cm}
\includegraphics[width=0.4\textwidth]{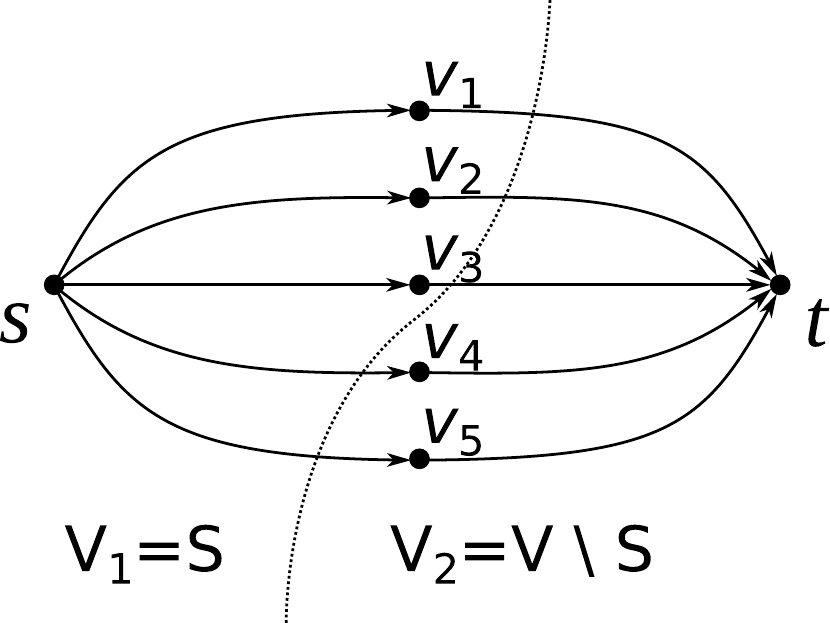}
\centering
\caption{Reduction: subset sum problem
is reduced to DGPAQ by creating a node for each $a_i$ (the nodes in the
center) and adding a source and sink node with edges as shown.\label{fig:subset_sum_reduction}}
\label{fig:npconstruction}
\vspace*{-.25cm}
\end{wrapfigure}

In this section, we show that the problem under consideration is
NP-complete when restricted to the case $k=2$ and $\epsilon=0$, and also hard to
approximate with a finite approximation factor for $k\geq3$.
A given solution for an instance of DGPAQ can be verified in linear time by
constructing the quotient graph $\mathcal{Q}$, checking the balance constraint
and checking $\mathcal{Q}$ for acyclicity. The last task can be done in linear
time in the size of $\mathcal{Q}$ using Kahn's
algorithm~\cite{kahn1962topological}. We now reduce the subset sum problem to
our problem. The proof is inspired by the reduction used
in~\cite{picard1980structure} which shows that the most balanced minimum cut
problem is NP-complete.

\begin{theorem}
The DGPAQ problem is NP-complete for the bi-partitioning case with $\epsilon=0$.
\end{theorem}

\begin{proof}
We reduce the NP-complete~\cite{gary1979computers} subset sum problem to DGPAQ\@.
The decision version of the subset problem is stated as follows:
Given a set of integers $\{a_1, \ldots, a_n\}$, is there a non-empty subset
$I \subseteq \{1, \ldots, n\}$
such that $\sum_{i \in I} a_i = \sum_{i \not \in I} a_i$ holds?
The construction of an equivalent instance of DGPAQ is as follows:
We construct a DAG $G=(V,E,c)$ with nodes $s,t \in V$ as well as a
node $v_i \in V$ for each $i \in \{1,\ldots, n\}$.
Then we set $A := \sum_i 2a_i$ and define the node weights as $c(s), c(t) := A$,
$c(v_i) := 2a_i$.
Afterwards, we insert edges $(s,v_i)\ \forall i$ and $(v_i,t)\ \forall i $.
The graph is a DAG -- an example topological ordering puts $s$ first, $t$ last and the remaining nodes at
arbitrary positions in between.
Figure~\ref{fig:npconstruction} illustrates the construction.
By definition, $L_{\max}=3A/2$ for this instance of DGPAQ\@. Note that by
construction $A$ is divisible by 2.
The construction can be done in polynomial time.
Note that all balanced partitions $(S,V \backslash S)$ cut $n$ edges, and due to
the balance constraint $s$ and $t$ can never be in the same block.
This ensures that there cannot be any edge $(u,v)$ with $u \in V\backslash S$ and $v \in S$
and hence the quotient graph is acyclic.
If the subset sum instance is a yes instance, then there is perfectly balanced
bipartition~and~vice~versa.
\end{proof}
The following theorem shows that it is not possible to find a finite factor
approximation algorithm for our general problem where $k$ is not a constant.
The proof is a modification of the proof by Andreev and
Räcke~\cite{andreev2006balanced} which shows this for the classical graph
partitioning problem, \ie no acyclicity constraint and for undirected inputs.
Hence, we follow the proof~of~\cite{andreev2006balanced}~closely with the
difference being that the inputs that we construct are DAGs.
\begin{theorem}
The directed graph partitioning problem with acyclic quotient graph has no
polynomial time approximation algorithm with a finite approximation factor
for $\epsilon=0, k\geq 3$ unless~$P=NP$.
\end{theorem}
\begin{proof}
The 3-Partition problem is defined as follows. Given $n=3k$ integers $a_1,
\ldots, a_n$ and a threshold $A$ such that $A/4 < a_i < A/2$ and $\sum_i
a_i = kA$, decide whether the numbers can be partitioned into triples such
that each triple adds up to $A$. This problem is \emph{strongly}
NP-complete~\cite{gary1979computers}, \ie the problem remains NP-complete
if all numbers $a_i$ and $A$ are polynomially~bounded.

Now suppose we have an approximation algorithm for the directed graph
partitioning problem with acyclic quotient graph for $\epsilon=0$. We can
use this algorithm to decide the 3-Partition problem with polynomially
bounded numbers. To do so, we construct a graph~$G$ that contains $n$
subgraphs.
Subgraph $i$ has $a_i$ nodes. All weights are set to $1$.
We make each of the subgraphs a directed clique, \ie all edges $(u,v)$ with $u
< v$ are inserted into the subgraph. By construction $G$ is a DAG\@.
 This is the main difference to
$\cite{andreev2006balanced}$ in which the subgraphs~are~undirected~cliques.
Also since all numbers are polynomially bounded, the construction takes
polynomial time.

Now, if the 3-Partition instance can be solved, the $k$-DGPAQ problem in $G$
can be solved without cutting any edge. Note that this solution also
fulfills the acyclicity constraint. If the 3-Partition instance cannot be
solved, then the optimum solution to the $k$-DGPAQ problem will cut at
least one edge. An approximation algorithm with finite approximation factor
has to differentiate between these two cases. Hence, it can solve the
3-Partition problem.
\end{proof}

\section{Heuristic Algorithms}
\label{heuristics}
In this section we present simple yet effective construction and local search
heuristics to tackle the problem. Our general approach is as follows: First
create an initial solution based on a topological ordering of the input graph and
then apply a local search strategy to improve the objective of the solution while
maintaining both constraints. We start the section with the construction
algorithm and then present different local~search~heuristics.

\subsection{Construction Algorithm}
All of our local search heuristics start with an initial partitioning that
fulfills both constraints, \ie the quotient graph is acyclic and the
balance constraint is satisfied.
Our algorithm does this by computing a random topological ordering of the nodes using a
modified version of Kahn's algorithm with randomized tie-breaking. More
precisely, the algorithm initializes a list S with all nodes that have indegree
zero and an empty list T. It then repeats the following steps until the list S
is empty: Select a node from S uniformly at random and remove it from the list.
Add the node to the tail of T. Remove all outgoing edges of the node. If this
reduces the indegree of another node to zero, add it to S. When the algorithm
terminates, the list T is a topological ordering of all nodes unless the graph
has a cycle.
Using list T, we can now derive initial solutions by dividing the graph into
blocks of consecutive nodes w.r.t to the ordering.
Due to the properties of the topological ordering there is no node in a block
$V_{j}$ that has an outgoing edge ending in a block $V_{i}$ with $i<j$. Hence,
the quotient graph of our solution is cycle-free. In addition, the blocks are
chosen such that the balance constraint is fulfilled.
There is obviously a large number of possible divisions. Our algorithm
generates a balanced initial partitioning by dividing the ordering
into blocks of size $\lfloor\frac{c(V)}{k}\rfloor$ or
$\lceil\frac{c(V)}{k}\rceil$ uniformly at random.
Since the construction algorithm is randomized, we run the heuristics $\ell$
times with different initial partitionings and pick the best~solution~afterwards.

\subsection{Local Search Algorithms}
Our local search heuristics take a given initial solution and move nodes
between the blocks in order to decrease the edge cut.
The reduction of the edge cut after a move is called the \emph{gain}
of the move. To compute the gain when moving node $v$, we define two functions:
\begin{align*}
C_{in}(v, i) &:= \omega(\{(u,v)\in E : u\in V_{i}\}) \\
C_{out}(v, i) &:= \omega(\{(v,u)\in E : u\in V_{i}\})
\end{align*}
\begin{wrapfigure}{r}{0.4\columnwidth}
\vspace*{-.5cm}
\includegraphics[width=0.35\columnwidth]{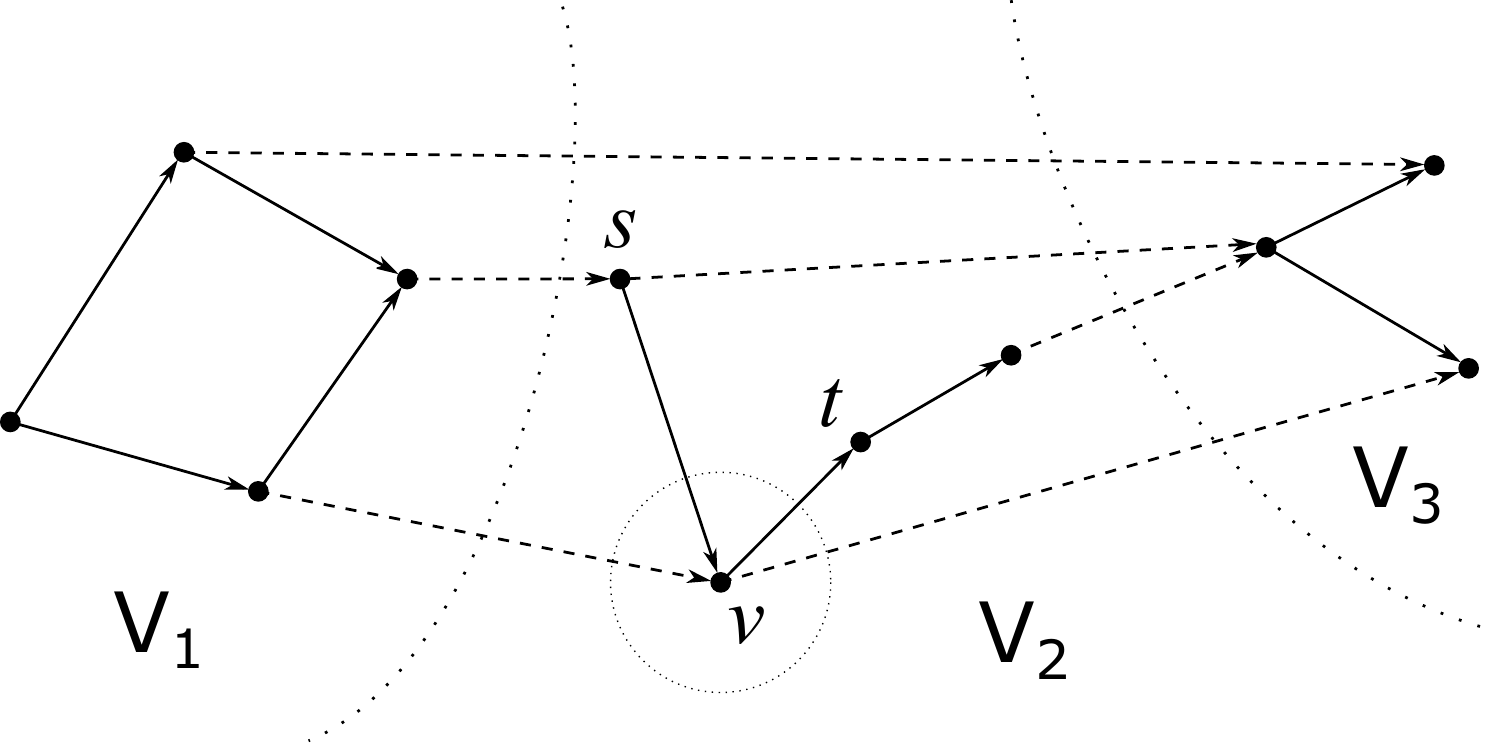}
\centering
\caption{A DAG divided into three blocks. Internal edges are solid,
external edges are dashed. Node $v$ is a node that has
non-zero internal and external cost for both $C_{in}$ and $C_{out}$. Because
of $(s, v)\in E \Rightarrow C_{in}(v, 2)>0$, the node cannot be moved to
$V_1$. Because of $(v, t)\in E \Rightarrow C_{out}(v, 2)>0$, the node cannot
be moved to $V_3$ either.\label{fig:partitioning_example}}
\vspace*{-.25cm}
\end{wrapfigure}
Roughly speaking, $C_{in}$ is the combined weight for all edges that start in
nodes of block $V_i$ and end in $v$.
Analogously, $C_{out}$ is the combined weight of all edges that start in $v$ and
connect to nodes in the block~$V_i$.
If~$v \in V_i$, these costs are the weights of \emph{internal} edges.
These edges will become external edges and increase the objective if we move $v$
to a different block.
If $v \in V_j, j \neq i$, then these costs are weights of \emph{external} edges,
which will become internal and thus reduce the edge cut if $v$ is moved to
$V_i$.
Figure~\ref{fig:partitioning_example} shows an example of internal and external
edges.

We have multiple local search heuristics that differ in the size of the local
search neighborhood: Simple Moves, Advanced Moves, Global Moves as well as FM moves.
We found that the heuristics can often yield better results with a different
initial partitioning. In order to compare the different heuristics, we will give
each heuristic the same time budget and will restart the heuristics for
different initial partitionings until it is exhausted.

\paragraph*{Simple Moves (SM)}
Simple moves start by picking a node $v$ and moving it to a different block if
this does not violate the constraints and improves the objective.
Our simple move heuristic only considers to move a node $v\in V_i$ to adjacent
blocks $V_{i-1}$ and $V_{i+1}$. This is because there is a fast algorithm to
check the \emph{acyclicity constraint}.
Assuming that the given solution is feasible with respect to both constraints,
it is sufficient to check whether $C_{out}(v,i)=0$ in the case that we want to
move $v$ to $V_{i+1}$ and $C_{in}(v,i)=0$ in the case that we want to move $v$
to $V_{i-1}$.
The gain of a node movement depends on the block and is calculated as:

\vspace*{-.5cm}
\begin{equation*}
\begin{cases}
C_{in}(v,i-1)-C_{out}(v,i) & \text{when moving } v \text{ to } V_{i-1}\\
C_{out}(v,i+1)-C_{in}(v,i) & \text{when moving } v \text{ to } V_{i+1}.
\end{cases}\label{eq:gain}
\end{equation*}

A block is eligible if the move does not create a cycle and does not overload
the block. In addition, the gain has to be positive or zero but the balance of
the partitioning is improved.
If there is such a block, we move $v$ to it. In the case that both blocks are
eligible for the move and have the same gain, the heuristic selects one
uniformly at random.

We repeat the process for all nodes. Our heuristic stops if there is no node
with positive gain or balance cannot be improved.
Hence, our heuristic terminates when a local minimum is found with respect to the
local search neighborhood defined above. Note that even though the edge cut is
not \emph{strictly} monotonically decreasing, the combination of edge cut and
difference in block weight is.
In one pass, the heuristic considers the in- and outgoing edges of all nodes.
Thus, each edge is considered exactly twice to calculate the gain for all nodes
and the complexity of the heuristic is $O(m)$ per round.

\paragraph*{Advanced Moves (AM)}
This algorithm increases the local search neighborhood of the Simple Moves
algorithm by considering more target blocks for a move. For the node $v\in V_i$
under consideration, all incoming edges are checked to find the node $u\in V_A$
where $A$ is maximal. Also all outgoing edges are checked to find the node $w\in
V_B$ where $B$ is minimal. Since the original partition was obtained from a
topological ordering, $A\leq i\leq B$ must hold, otherwise there would be back
edges in the ordering and thus it would not be a topological ordering. If
$A = i = B$, then the node $v$ has in- and outgoing edges in its own block and
cannot be moved. If $A < i$, then the node can be moved to blocks preceding
$V_i$ up to and including $V_A$ in the topological ordering without creating a
cycle. This is because all incoming edges of the node will either be internal to
block $V_A$ or are forward edges starting from blocks preceding $V_A$.
Therefore it is still a topological ordering. However, when the node is moved to
a block preceding $V_A$, the edge starting in this block becomes a back edge
and the ordering is not a topological ordering anymore.
Similar, if $i < B$, the node can be moved to blocks succeeding $V_i$ up to and
including $V_B$.
Thus moving the node to $V_j$ with $j\in \{A,\ldots,B\}\setminus \{i\}$ will
preserve the topological ordering of blocks. This is a sufficient condition to
ensure the acyclicity constraint and is not computationally expensive to check.
However, since it is not a necessary condition, it might prevent the heuristic
from testing some possible moves. The Global Moves heuristic does not have this
limitation, but has a higher~computational~complexity.

The gain of the moves to all allowed $V_j$ is computed with the cost functions
described in the previous section as
$C_{in}(v, j)-C_{out}(v,i)+C_{out}(v,j)-C_{in}(v,i)$.
In each iteration, the move with the largest gain such that the constraints are
maintained is selected. Tie-breaking and gains of zero are handled in the same
way as in Simple Moves.

This heuristic considers each edge exactly twice in order to calculate the gain
when moving the node to any other block.
Afterwards, a block yielding maximal gain is selected, which can be done in
time proportional to the degree of a node.
Thus, the complexity of this heuristic is $O(m)$.

\paragraph*{Global Moves (GM)}
With this algorithm, we increase the local search neighborhood even further by
considering all other blocks.
Starting from the initial partition, the algorithm computes the adjacency lists
of the quotient graph. Throughout the algorithm the quotient graph is kept
up-to-date.
When moving a node we update the adjacency information of the quotient graph and
record whether a new edge has been created. If this is the case we check the
quotient graph for acyclicity by using Kahn's algorithm and undo the last
movement if it created a cycle.
\csch{add a sentence with intuition}
\csch{missing analysis of all algorithms, faster check by keeping ordering of quotient graph up to date?, FM-based stuff?}

The calculation of the gain values can be done in $O(m)$ as for the other
heuristics. For a node, the heuristic needs to check the acyclicity constraint
for all considered moves/blocks in the worst case.
Since Kahn's algorithm checks the quotient graph for acyclicity, the total
complexity of this heuristic is $O(m(m_\mathcal{Q}+k))$ where $m_\mathcal{Q}$
is the number of edges in the quotient graph. If the quotient graph is sparse,
\ie $m_Q$ is $O(k)$, we get a complexity~of~$O(km)$.

\paragraph*{FM Moves (FM)}
This heuristic combines the quick check for acyclicity of the Advanced Moves
heuristic with an adapted Fiduccia-Mattheyses algorithm \cite{fiduccia1982lth}
which gives the heuristic the ability to climb out of a local minimum.
The initial partitioning is improved by exchanging nodes between a pair of
blocks even if the gain is negative. The partition with the best objective
that was seen during the pass will be returned.
A pass starts with two blocks $A$ and $B$, where $A$ precedes $B$ in the
topological ordering of blocks. The algorithm will then calculate the gain for
moving \emph{enabled} boundary nodes to the other block.
Using the same criterion to guarantee acyclicity as the Advanced Moves heuristic, we
say that a boundary node is enabled if it is in $A$ and does not have outgoing
edges to nodes that precede $B$ or it is in $B$ and does not have incoming
edges from nodes that follow $A$.
The candidate moves, consisting of a gain and a node identifier, are inserted
into a priority queue.
The queue is a binary heap where the total order on the elements
is implemented by comparing the gain of the moves and, if the gain is the same,
a random number that is generating upon insertion.

In a loop that runs until the priority queue is depleted, the first move is
extracted from the queue. If the selected move would overload the target block
or is not enabled because it was disabled in a previous loop iteration, the
heuristic continues with the next iteration.
Otherwise, the move will be committed even if the gain is negative. The node is
then locked, \ie it cannot be moved again during this pass. This prevents
thrashing and guarantees the termination of the algorithm.
Unlike the Fiduccia-Mattheyses algorithm, a move in this scenario does not
change the gain, it disables and enables other moves. For example, if a node $w$
is moved from $A$~to~$B$, the heuristic will disable all nodes $v$ in block $B$
with $(w,v) \in E$ since they do not fulfill the condition for acyclicity
anymore and moving any of them to $A$ would introduce a back edge in the
topological ordering of blocks. This does not necessarily mean that the quotient
graph would become cyclic, however, assuring this would require a more expensive
check like Kahn's algorithm.
Note that the gain of the moves does not need to be re-calculated since $w$ was
locked and thus all nodes $v$ will not be enabled again in this pass.
On the other hand, moving $w$ enables nodes in $A$ if they are connected with
an outgoing edge to $w$ and if after the move they do not have other outgoing
edges to blocks preceding $B$.
The heuristic will calculate the gain for these nodes, enable and insert them
into the priority queue. A move from $B$ to $A$ will enable and disable moves
correspondingly.
The loop will continue to move nodes between the blocks until the priority queue
is depleted, which occurs when all nodes are either disabled or locked.
Since the number of loop iterations is hard to predict due to the reinsertion of
moves, it is limited to $2n/k$ which did not have a measurable impact on the
quality of obtained partitionings.
The best objective that was achieved in the pass is recorded. In the final step,
the last moves are undone if required to reach the corresponding partitioning.
This terminates the inner pass of the heuristic.

The outer pass of the heuristic will repeat the inner pass for randomly chosen
pairs of blocks. At least one of these blocks has to be ``active''. Initially, all blocks
are marked as ``active''. If and only if the inner pass results in movement of
nodes, the two blocks will be marked as active for the next iteration. The
heuristic stops if there are no more~active~blocks.

The overall time to compute gain values is $O(m)$.
We now analyse the running time for a pair of blocks. In the worst case, all
nodes of both blocks are enabled in the beginning and initializing the priority
queue with $2n/k$ nodes requires $O(\frac{n}{k})$ time.
Note that we cannot use a bucket priority queue, since the weights associated
with the edges can be more or less arbitrarily distributed.
Removing a node with the best gain from the queue takes $O(\log \frac{n}{k})$
time.
If a move is committed in an iteration, the heuristic needs to calculate the
gain of adjacent nodes. However, the heuristic will never calculate the gain
of a move twice during a pass.
Thus the total complexity of the inner pass is~$O(\frac{n}{k}\log
\frac{n}{k})$.
Note that the inner pass needs to be performed for all pairs of
blocks which yields overall time $O(m+m_\mathcal{Q}\frac{n}{k}\log
\frac{n}{k})$ per round of the algorithm, or $O(m+n\log \frac{n}{k})$ if the
quotient graph is sparse.
\section{Experimental Evaluation}
\label{evaluation}
\label{s:experiments}
In this section we evaluate the performance of our algorithms. We start
by presenting methodology and the systems we use for the evaluation. Then we evaluate the solution quality on small instances
by comparing with the optimal solutions and evaluate our algorithms on complex imaging filters. We finish with testing the scalability of our algorithms.

\subparagraph*{Methodology.}
We have implemented the algorithms described above using C++.
All programs have been compiled using g++ 4.8.0 and 32 bit index data types.
The system we use is equipped with two Intel Xeon X5670 Hexa-Core processors
(Westmere) running at a clock speed of 2.93 GHz. The machine has 128GB main
memory, 12MB L3-Cache and 6$\times$256KB L2-Cache.
All instances described in this section will be made available~on~request.
\subparagraph*{Comparison with Optimal Solutions.}
This section compares the results of our heuristics against the optimal solution
obtained by a non-polynomial time algorithm that performs an exhaustive search.
We create a set of random graphs that are close to instances from typical
\begin{wraptable}{r}{0.6\columnwidth}
\centering
\vspace*{-.5cm}
\begin{tabular}{cr|r|r|r|r}
$k$ & $\epsilon$ & SM & AM & GM & FM\tabularnewline
\hline
\hline
\multirow{4}{*}{2} & 20 \% & 3.41 \%  & 3.41 \%  & 3.41 \%  & 0.26 \% \tabularnewline
\cline{2-6}
                   & 30 \% & 11.94 \% & 11.91 \% & 11.90 \% & 0.33 \% \tabularnewline
\cline{2-6}
                   & 40 \% & 14.71 \% & 14.78 \% & 14.58 \% & 1.29 \% \tabularnewline
\cline{2-6}
                   & 50 \% & 23.32 \% & 23.36 \% & 23.04 \% & 1.21 \% \tabularnewline
\hline
\multirow{4}{*}{4} & 20 \% & 1.89 \%  & 1.27 \%  & 1.33 \%  & 0.74 \% \tabularnewline
\cline{2-6}
                   & 30 \% & 4.03 \%  & 3.22 \%  & 3.25 \%  & 0.67 \% \tabularnewline
\cline{2-6}
                   & 40 \% & 5.09 \%  & 3.65 \%  & 3.69 \%  & 0.44 \% \tabularnewline
\cline{2-6}
                   & 50 \% & 6.50 \%  & 4.04 \%  & 4.19 \%  & 0.31 \% \tabularnewline
\end{tabular}
\vspace*{.25cm}
\caption{Each cell shows the averaged result of the heuristic for the current combination of block count $k$ and imbalance $\epsilon$. The value is the increase in cost compared to the optimal solution.\label{tab:random_graphs}}
\vspace*{-.5cm}
\end{wraptable}
applications. Our generation algorithm works by consecutively adding new graph
levels with a random number of nodes. Each of the new nodes is connected to a
random number of nodes in previous levels.
Because the application domain of this work is imaging, we use a small number of
input and output nodes (between one and three) which is typically the case for
imaging and vision kernels (compare library of OpenVX vision functions
\cite{openvx_vision_functions}).
Since the weight of nodes is representing the program size, we select a random
value between the size of the smallest and the largest kernel in an implementation
of the Local Laplacian filter for our target platform.
The weight of edges is uniformly chosen between 1 and 100 to account for
different sizes for intermediate buffers between the functions.

Because the following parameters have a major impact on the structure of
the graph, we use two different values for each and generate 25 graphs
for each of the eight resulting parameter combinations:
\begin{itemize}
\item The maximum size of a graph level is either set to a high value ($\sqrt{n}$) which results in a graph that can in extreme cases have $\sqrt{n}$ levels with about $\sqrt{n}$ nodes each, meaning that there is a high amount of data parallelism, and low values ($\sqrt[4]{n}$) such that the graph resembles more a long chain of nodes and thus represents the classical imaging pipeline with low data parallelism on kernel level.
\item The maximum number of edges is either set to the lowest number that ensures that inner nodes have at least one incoming and one outgoing edge and that the graph is connected or to $\sqrt{n}$ per node such that the number of edges scales with the problem size. This reflects applications with few and many data dependencies between functions.
\item The maximum distance in terms of node indices, over which new nodes are connected to preceding nodes in the graph, is either set to a low value that results in a graph where nodes only have incoming edges from the closest preceding levels or it is set to $n$ which means that there is no restriction on where edges can start. The first case models application where data is short-lived and only needed for the next step in a pipeline while the second case represents scenarios with a long data lifetime.
\end{itemize}

These $200$ different problems instances were generated for problem sizes in
the range of $n \in [10,\ldots,20]$ nodes each.
Table~\ref{tab:random_graphs} shows the averaged approximation factor of the
four heuristics when using a time budget of 10 milliseconds. Running times of
the exhaustive search algorithm vary between a couple of seconds for the small
instances and a couple of days for the large instances.
The results show a good approximation of the optimal result that improves
further with increasing size of the search by the different heuristics.
The quality of SM, AM and GM degrades with large $\epsilon$ since they can get
trapped in a local minimum, FM moves on the other hand shows a close and
consistent approximation.

\vspace*{-.5cm}
\subparagraph{Local Laplacian Filter.}

The Local Laplacian filter is an edge-aware image processing filter.
A detailed description of the algorithm and theoretical background
is given in~\cite{paris2011local}.
\begin{wraptable}{r}{0.6\columnwidth}
\centering
\begin{tabular}{l|r|r|r|r|r|}
 & man. & SM & AM & GM & FM\tabularnewline
\hline
\hline
number programs & \numprint{20} & \numprint{22} & \numprint{16} & \numprint{16} & \numprint{17}\tabularnewline
\hline
number gangs & \numprint{5} & \numprint{7} & \numprint{4} & \numprint{4} & \numprint{6}\tabularnewline
\hline
1-level edge cut & \numprint{11.4} & \numprint{10.9} & \numprint{8.8} & \numprint{8.9} & \numprint{10.4}\tabularnewline
\hline
2-level edge cut & \numprint{8.9} & \numprint{6.2} & \numprint{5.0} & \numprint{4.7} & \numprint{6.1}\tabularnewline
\hline
rel.\ execution time & \numprint{1.00} & \numprint{1.09} & \numprint{0.89} & \numprint{1.04} & \numprint{1.31}\tabularnewline
\end{tabular}
\vspace*{.25cm}
\caption{Table comparing the results of the manual implementation with the solution found by the heuristic\label{tab:laplace}}
\vspace*{-.5cm}
\end{wraptable}
The algorithm uses concepts of \emph{Gaussian pyramids} and \emph{Laplacian pyramids} as well as a point-wise remapping
function in order to enhance image details without creating artifacts.
We model the data flow of the filter as a DAG where nodes represent simple
function primitives, \eg{}upsampling, downsampling and gaussian filtering for
both image dimensions for each level of the pyramid generation.
If there is a direct data dependency between two nodes, they are connected by an
edge with a weight of the number of pixels of the corresponding buffer.
The node weight is set to the program memory used by the primitive.
The DAG has $72$ nodes and $93$ edges in total in our configuration.
In an existing implementation, the primitives were grouped in a functional way
(\eg{}pyramid generation) by the developer into programs and then assigned to
a total of five scheduling gangs.
To evaluate the heuristics, we use a first pass with $L_{\max}$ set to the size
of the program memory to find a good composition of function primitives into
programs.
The resulting quotient graph is then used in a second pass where $L_{\max}$ is
set to the total number of PEs in order to find scheduling gangs that minimize
external memory transfers. In this second step the acyclicity constraint is
crucially important.
In both passes, empty blocks were explicitly permitted to allow the heuristics
to reduce the number of gangs.
The time budget given to each heuristic is one minute.
We also found that due to (desired) compiler optimizations, the final program
memory size of a program can be smaller than the sum of its primitives. Since the
entire process of partitioning, code generation and compilation is automated, we
took advantage of this by slowly increasing $\epsilon$ until the programs became
too large. The results are shown in Table~\ref{tab:laplace}. The 1-level edge
cut shows the amount of communication between programs, the 2-level edge cut
between gangs, both in megapixels. The cycle count was obtained with a cycle-true
compiled simulator of the hardware platform. Since gangs are executed one after
another, the longest execution time of a program in a gang limits the throughput.
The table shows the summed execution of each limiting program relative to the
manual implementation, not taking bandwidth limitations into account in order to
see when a schedule will be compute-limited.
All heuristics improve the edge cut by at least $30 \%$, thus the schedule will
be superior on all platforms where the manual implementation is
bandwidth-limited. In addition, by reducing edge cut, the AM and GM heuristics
find partitionings that require fewer gangs. For AM, this improves execution time,
so the schedule will be superior even if bandwidth is not the limiting factor.
For GM, the heuristic makes an additional choice that reduces edge cut even
further, but does not balance compute-intensive programs well and thus schedule
time is not improved. This is not addressed in our graph partitioning heuristic
and should be considered in a scheduling algorithm.
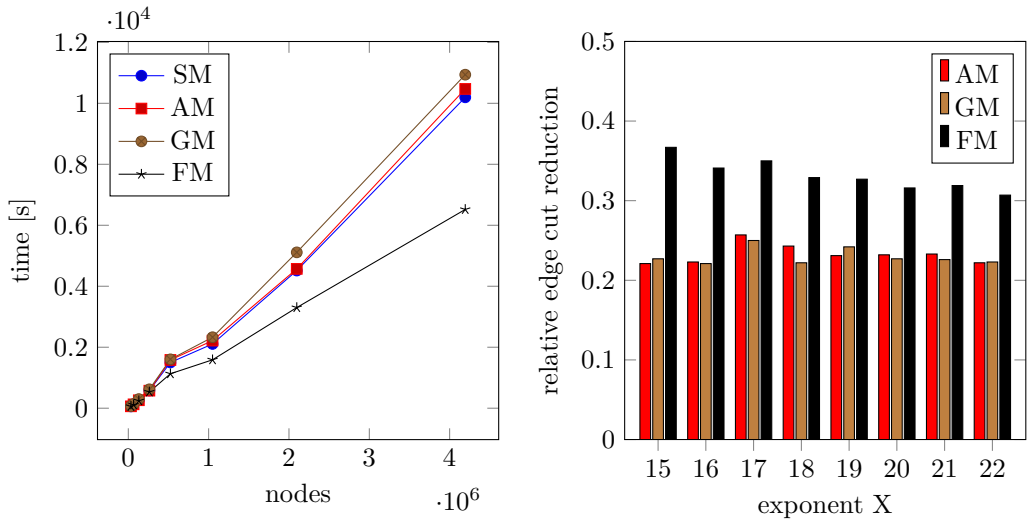
\begin{figure}[t]
\centering
\begin{subfigure}{.49\textwidth}
  \begin{tikzpicture}
\pgfplotsset{
    width  = \columnwidth,
    height = \columnwidth,
    legend pos = north west,
}
\begin{axis}[
    ylabel             = {time [s]},
    xlabel             = {nodes},
]
\addplot coordinates {
	(32759,     57.962   )
	(65532,     127.100  )
	(131067,    261.790  )
	(262141,    578.702  )
	(524280,    1500.500 )
	(1048572,   2106.960 )
	(2097142,   4514.620 )
	(4194299,   10195.800)
};
\addplot coordinates {
	(32759,     62.110   )
	(65532,     127.227  )
	(131067,    260.960  )
	(262141,    575.048  )
	(524280,    1582.250 )
	(1048572,   2217.540 )
	(2097142,   4569.460 )
	(4194299,   10471.100)
};
\addplot coordinates {
	(32759,     68.621   )
	(65532,     139.009  )
	(131067,    299.896  )
	(262141,    620.418  )
	(524280,    1603.560 )
	(1048572,   2328.400 )
	(2097142,   5111.460 )
	(4194299,   10938.500)
};
\addplot coordinates {
	(32759,     57.904  )
	(65532,     109.336 )
	(131067,    242.443 )
	(262141,    532.275 )
	(524280,    1130.080)
	(1048572,   1587.100)
	(2097142,   3300.750)
	(4194299,   6518.600)
};
\legend{SM, AM, GM, FM}
\end{axis}
\end{tikzpicture}
  \vspace*{.22cm}
  \caption{\centering execution time averaged over 100 passes}
\end{subfigure}
\begin{subfigure}{.49\textwidth}
  \vspace*{.22cm}
  \begin{tikzpicture}
\pgfplotsset{
    width  = \columnwidth,
    height = \columnwidth,
    legend pos = north east,
}
\pgfplotsset{
        /pgfplots/ybar legend/.style={
        /pgfplots/legend image code/.code={
        \draw[##1,/tikz/.cd,bar width=3pt,yshift=-0.2em,bar shift=0pt]
                plot coordinates {(0cm,0.8em)};},},
}
\begin{axis}[
    bar width         = {4pt},
    ybar              = {2*\pgflinewidth},
    ylabel            = {relative edge cut reduction},
    ymin              = {0},
    ymax              = {0.5},
    xlabel            = {exponent X},
    xtick             = {data},
    xtick pos         = {left},
    symbolic x coords = {15,16,17,18,19,20,21,22},
    enlarge x limits  = {0.1},
]
\addplot[ybar, fill=red] coordinates {
	(15, 0.221)
	(16, 0.223)
	(17, 0.257)
	(18, 0.243)
	(19, 0.231)
	(20, 0.232)
	(21, 0.233)
	(22, 0.222)
};
\addplot[ybar, fill=brown] coordinates {
	(15, 0.227)
	(16, 0.221)
	(17, 0.250)
	(18, 0.222)
	(19, 0.242)
	(20, 0.227)
	(21, 0.226)
	(22, 0.223)
};
\addplot[ybar, fill=black] coordinates {
	(15, 0.367)
	(16, 0.341)
	(17, 0.350)
	(18, 0.329)
	(19, 0.327)
	(20, 0.316)
	(21, 0.319)
	(22, 0.307)
};
\legend{AM, GM, FM}
\end{axis}
\end{tikzpicture}
  \caption{\centering edge cut relative to Simple Moves}
\end{subfigure}
\caption{Graph showing the execution time of each heuristic and the relative edge cut on directed random geometric graphs rggX. }
\label{tab:large_graphs}
\end{figure}

\subparagraph*{Random Geometric Graphs.}
We now look at the scalability of our heuristics. We do this
on \emph{random geometric graphs} where nodes represent random points in the unit square and edges
connect nodes whose Euclidean distance is below $0.55 \sqrt{ \ln n / n }$.
This threshold was chosen in order to ensure that the graph is almost connected.
These graphs were taken from~\cite{benchmarksfornetworksanalysis} and were initially undirected.
We convert them into DAGs by directing edges from smaller to larger node ids.
The graph rgg$X$ has $2^X$ nodes. We vary $X \in [15,\ldots,22]$.
The allowed imbalance was set to $3 \%$ since this is one of the values used in \cite{walshaw2000mpm}.
Figure~\ref{tab:large_graphs} shows the averaged time
required for 100 passes of each heuristic and the relative improvement in edge
cut that was found for $k=8$ by the more advanced heuristics in comparison to
the Simple Moves heuristic.
The figure shows a linear growth of our heuristics respective to number of nodes.
The worst case complexity of FM moves was shown to be superlinear since it had
to be assumed that all nodes are boundary nodes, which is not the case here. We conclude that our algorithms scale well to large problems.

\section{Conclusion}
\label{conclusion}
In this work we designed, implemented and evaluated new heuristics that
partition application graphs under constraints that are important for
multiprocessor scheduling. It was shown that the constrained problem is
NP-complete and that the heuristics yield good approximations of the optimal
solution for small problem instances and have a linear growth for larger
instances. In a simulation we could show the positive impact on communication
volume of a real application for all heuristics and in one case even a
reduction of execution time when bandwidth is not the limiting factor due to
a reduction in number of scheduling gangs.
The running time of the heuristics w.r.t to problem size is sufficient for this
application domain, especially since the algorithms only need to run at
compile-time.
Also the communication volume was reduced by an extent that suggests that for
future work it will be more rewarding to introduce additional objectives that
help to improve gang execution time by better balancing compute-intensive
kernels.

\bibliographystyle{plain}
\bibliography{bibliography,phdthesiscs}

\begin{thebibliography}{10}

\bibitem{Karypis06}
A.~Abou-Rjeili and G.~Karypis.
\newblock {Multilevel Algorithms for Partitioning Power-Law Graphs}.
\newblock In {\em Proc. of 20th IPDPS}, 2006.

\bibitem{andreev2006balanced}
K.~Andreev and H.~R{\"a}cke.
\newblock {Balanced Graph Partitioning}.
\newblock {\em Theory of Computing Systems}, 39(6):929--939, 2006.

\bibitem{benchmarksfornetworksanalysis}
D.~A. Bader, H.~Meyerhenke, P.~Sanders, C.~Schulz, A.~Kappes, and D.~Wagner.
\newblock {Benchmarking for Graph Clustering and Partitioning}.
\newblock In {\em Encyclopedia of Social Network Analysis and Mining}, to
  appear.

\bibitem{GPOverviewBook}
C.~Bichot and P.~Siarry, editors.
\newblock {\em Graph Partitioning}.
\newblock Wiley, 2011.

\bibitem{SPPGPOverviewPaper}
A.~Bulu\c{c}, H.~Meyerhenke, I.~Safro, P.~Sanders, and C.~Schulz.
\newblock {Recent Advances in Graph Partitioning}.
\newblock In {\em Algorithm Engineering -- Selected Topics, \emph{to app.},
  ArXiv:1311.3144}, 2014.

\bibitem{ptscotch}
C.~Chevalier and F.~Pellegrini.
\newblock {PT-Scotch}.
\newblock {\em Parallel Computing}, 34(6-8):318--331, 2008.

\bibitem{feitelson1992gang}
D.~G Feitelson and L.~Rudolph.
\newblock Gang scheduling performance benefits for fine-grain synchronization.
\newblock {\em Journal of Parallel and distributed Computing}, 16(4):306--318,
  1992.

\bibitem{fiduccia1982lth}
C.~M. Fiduccia and R.~M. Mattheyses.
\newblock {A Linear-Time Heuristic for Improving Network Partitions}.
\newblock In {\em Proceedings of the 19th Conference on Design Automation},
  pages 175--181, 1982.

\bibitem{gary1979computers}
M.~R. Gary and D.~S. Johnson.
\newblock Computers and intractability: A guide to the theory of
  np-completeness, 1979.

\bibitem{goossens2016optimal}
J.~Goossens and P.~Richard.
\newblock Optimal scheduling of periodic gang tasks.
\newblock {\em Leibniz transactions on embedded systems}, 3(1):04--1, 2016.

\bibitem{openvx}
Khronos Group.
\newblock {The OpenVX API}.
\newblock \url{https://www.khronos.org/openvx/}.

\bibitem{openvx_vision_functions}
Khronos Group.
\newblock {The OpenVX Specification: Vision Functions}.
\newblock
  \url{https://www.khronos.org/registry/OpenVX/specs/1.0/html/da/db6/group__group__vision__functions.html}.

\bibitem{kahn1962topological}
A.~B. Kahn.
\newblock Topological sorting of large networks.
\newblock {\em Communications of the ACM}, 5(11):558--562, 1962.

\bibitem{karypis1998fast}
G.~Karypis and V.~Kumar.
\newblock {A Fast and High Quality Multilevel Scheme for Partitioning Irregular
  Graphs}.
\newblock {\em SIAM Journal on Scientific Computing}, 20(1):359--392, 1998.

\bibitem{lee1987synchronous}
E.~A. Lee and D.~G. Messerschmitt.
\newblock Synchronous data flow.
\newblock {\em Proceedings of the IEEE}, 75(9):1235--1245, 1987.

\bibitem{meyerhenke2006accelerating}
H.~Meyerhenke, B.~Monien, and S.~Schamberger.
\newblock {Accelerating Shape Optimizing Load Balancing for Parallel FEM
  Simulations by Algebraic Multigrid}.
\newblock In {\em Proc. of 20th IPDPS}, 2006.

\bibitem{panda2001data}
P.~R. Panda, F.~Catthoor, N.~D. Dutt, K.~Danckaert, E.~Brockmeyer, C.~Kulkarni,
  A.~Vandercappelle, and P.~G. Kjeldsberg.
\newblock Data and memory optimization techniques for embedded systems.
\newblock {\em ACM Transactions on Design Automation of Electronic Systems
  (TODAES)}, 6(2):149--206, 2001.

\bibitem{paris2011local}
S.~Paris, S.~W. Hasinoff, and J.~Kautz.
\newblock Local laplacian filters: edge-aware image processing with a laplacian
  pyramid.
\newblock {\em ACM Trans. Graph.}, 30(4):68, 2011.

\bibitem{Scotch}
F.~Pellegrini.
\newblock {Scotch Home Page}.
\newblock {\url{http://www. labri.fr/pelegrin/scotch}}.

\bibitem{picard1980structure}
J.~C. Picard and M.~Queyranne.
\newblock {On the Structure of All Minimum Cuts in a Network and Applications}.
\newblock {\em Mathematical Programming Studies}, 13:8--16, 1980.

\bibitem{rainey2014addressing}
E.~Rainey, J.~Villarreal, G.~Dedeoglu, K.~Pulli, T.~Lepley, and F.~Brill.
\newblock Addressing system-level optimization with openvx graphs.
\newblock In {\em Proceedings of the IEEE Conference on Computer Vision and
  Pattern Recognition Workshops}, pages 644--649, 2014.

\bibitem{kaffpa}
P.~Sanders and C.~Schulz.
\newblock {Engineering Multilevel Graph Partitioning Algorithms}.
\newblock In {\em Proc. of the 19th European Symp. on Algorithms}, volume 6942
  of {\em LNCS}, pages 469--480. Springer, 2011.

\bibitem{schloegel2000gph}
K.~Schloegel, G.~Karypis, and V.~Kumar.
\newblock {Graph Partitioning for High Performance Scientific Simulations}.
\newblock In {\em The Sourcebook of Parallel Computing}, pages 491--541, 2003.

\bibitem{Sou35}
R.~V. Southwell.
\newblock {Stress-Calculation in Frameworks by the Method of {``Systematic
  Relaxation of Constraints''}}.
\newblock {\em Proc. of the Royal Society of London}, 151(872):56--95, 1935.

\bibitem{stavrinides2016scheduling}
G.~L. Stavrinides and H.~D. Karatza.
\newblock Scheduling different types of applications in a saas cloud.
\newblock In {\em Proceedings of the 6th International Symposium on Business
  Modeling and Software Design (BMSD’16)}, pages 144--151, 2016.

\bibitem{walshaw2000mpm}
C.~Walshaw and M.~Cross.
\newblock {Mesh Partitioning: A Multilevel Balancing and Refinement Algorithm}.
\newblock {\em SIAM Journal on Scientific Computing}, 22(1):63--80, 2000.

\bibitem{Walshaw07}
C.~Walshaw and M.~Cross.
\newblock {JOSTLE: Parallel Multilevel Graph-Partitioning Software -- An
  Overview}.
\newblock In {\em {Mesh Partitioning Techniques and Domain Decomposition
  Techniques}}, pages 27--58. 2007.

\end{thebibliography}

\begin{appendix}
\end{appendix}

\end{document}